\documentclass[aps, prx, reprint, amsmath,amssymb,showpacs,floatfix,longbibliography, twocolumn, superscriptaddress,nofootinbib]{revtex4-2}
\usepackage{times} 
\usepackage{graphicx}
\usepackage{dcolumn}
\usepackage{bm}
\usepackage{color}
\usepackage{xcolor}
\usepackage{hyperref}
\usepackage{enumitem}
\usepackage{cancel}

\usepackage{tikz-cd}

\usepackage{comment}
\hypersetup{colorlinks=true,citecolor=blue,linkcolor=blue, urlcolor=blue}
\hypersetup{linktocpage}
\newcolumntype{M}[1]{>{\centering\arraybackslash}m{#1}}
\newcolumntype{N}{@{}m{0pt}@{}}
\usepackage{environ}
\usepackage{MnSymbol}
\usepackage{CJKutf8}
\usepackage{footmisc}
\usepackage[bb=libus]{mathalpha}
\usepackage{makecell}

\bibpunct{[}{]}{,}{n}{}{}

\usepackage{amsfonts,amssymb,amsmath}
\usepackage[T1]{fontenc}
\usepackage{tikz}

\let\originalleft\left
\let\originalright\right
\renewcommand{\left}{\mathopen{}\mathclose\bgroup\originalleft}
\renewcommand{\right}{\aftergroup\egroup\originalright}

\newcommand{\zw}[1]{{#1}}

\usepackage{amsthm}

\newtheorem{theorem}{Theorem}

\theoremstyle{definition}
\newtheorem{definition}{Definition}[section]

\newtheorem{example}{Example}

\NewEnviron{eqs}{%
\begin{equation}\begin{split}
    \BODY
\end{split}\end{equation}
}

\begin{document}
\begin{CJK*}{UTF8}{gbsn}
\title{A framework for low-overhead quantum fault tolerance via spacetime lifting}
\author{Yijia Xu (许逸葭)}
\email{yijia@terpmail.umd.edu}
\affiliation{Joint Center for Quantum Information and Computer Science, University of Maryland, College Park,
Maryland 20742, USA}

\author{Yixu Wang}
\email{wangyixu@simis.cn}
\affiliation{Shanghai Institute for Mathematics and Interdisciplinary Sciences (SIMIS), Shanghai 200433, China}

\author{Zi-Wen Liu}
\email{zwliu0@tsinghua.edu.cn}
\affiliation{Yau Mathematical Sciences Center, Tsinghua University, Beijing, 100084, China }

\date{\today}

\begin{abstract}
Fault-tolerant quantum computation is inherently a spacetime problem, requiring not merely good static quantum error-correcting codes but also low-overhead protocols for protecting and manipulating encoded quantum information over time. Fault complexes provide a homological formalism for treating such protocols as single spacetime objects. Here we initiate the study of low-overhead fault complexes by introducing \emph{spacetime lifting}, a method that constructs fault complexes from symmetry-reduced product structures beyond standard foliation. We show that spacetime lifting yields fault complexes and in particular  memory experiments  with almost-linear fault distance in the total spacetime cost, which substantially outperforms existing constructions. We further endow fault complexes with the operational interpretation of measurement-based cluster state protocols and identify general conditions under which they realize fault-tolerant logical teleportation, showing that spacetime-lifted constructions combine favorable parameter scaling with operational realizations. Our work establishes a systematic route towards more efficient quantum fault tolerance through general complex constructions.

\end{abstract}
\maketitle
\end{CJK*}

\section{Introduction}

Scalable fault-tolerant quantum computation relies on quantum error-correcting codes, yet good static code properties such as high code rates and distances are not enough: one must design robust protocols for dynamically protecting and manipulating encoded quantum information in spacetime~\cite{gottesman2013fault}.
With recent experimental progress toward large-scale quantum error correction~\cite{google2025quantum,bluvstein2026fault}, the study of fault-tolerant protocols that explicitly account for both spatial and temporal overheads has only recently begun to emerge~\cite{yamasaki2024time,xu2024constant,homo_measurement,xu2025batched,zhou2025low,tan2025single,zheng2025high,he2025extractors,baspin2025fast,tamiya2026fault,williamson2026low,cowtan2026parallel,chang2026constant}.

Recent work introduced the \emph{fault complex} as an algebraic formalism for describing spacetime fault-tolerant processes~\cite{vasmer2025fault}, with applications to memory experiments, logical teleportation, syndrome extraction, GHZ-state preparation, foliated fault-tolerant schemes, and lattice surgery~\cite{RBH,raussendorf2006fault,raussendorf2007topological,nickerson2018measurement,newman2020generating,ftcomplex,vasmer2025fault}. This perspective provides a concise and unified language for analyzing dynamical fault-tolerant protocols by treating the entire spacetime process as a single geometric and algebraic object. In this spacetime picture, the central parameter that captures the actual robustness of the entire fault-tolerant process is the \emph{fault distance}, which is the minimum number of nontrivial undetectable spacetime faults. 
 This conceptual advance from the static code distance is crucial because faults can propagate through gates and measurements into correlated errors whose effect depends on where and when they occur, reflecting dynamical features beyond the distance of the underlying static codes.

Most existing approaches to fault-tolerant protocols start from explicit circuit constructions, including repeated syndrome extraction (foliation) \cite{RBH,foliation,roberts2020foliation,vasmer2025fault}, generalizations based on 3D cellular complexes \cite{nickerson2018measurement,newman2020generating,ftcomplex}, and spacetime circuit constructions \cite{bacon2017sparse,gottesman2022opportunities,delfosse2023spacetime,pesah2025fault,li2025spacetime} and their fault-tolerant ZX transformations \cite{ZXFT1,ZXFT2}. In this work, we adopt a broader and more fundamental perspective: rather than deriving fault complexes from circuits, we construct them directly from general chain complexes with desired properties. We show each fault complex naturally gives a fault-tolerant protocol in a bipartite cluster state, which can also be viewed a quantum process in a twisted spacetime. 
Specifically, we introduce a framework for constructing fault complexes based on the lifted product~\cite{panteleev2021quantum} which we call spacetime lifting.
 We show that this approach yields fault-tolerant protocols whose fault distances scale almost linearly with the total spacetime cost defined as the spacetime volume required to implement the protocol, which is nearly optimal and breaks the square-root fault distance scaling barrier of known foliated constructions \cite{foliation,vasmer2025fault}.
 {Furthermore, we develop the operational interpretation of fault complexes in the measurement-based setting and identify the conditions under which a fault complex realizes logical teleportation. Notably, the product structure endows these constructions with clear operational utility in terms of fault-tolerant MBQC and logical teleportation, naturally identifying input and output portals together with the logical correlations connecting them. } In doing so, we provide a general algebraic characterization of the conditions required for logical teleportation, extending the framework beyond foliated fault complexes to general chain complexes.

\section{Fault complexes from code complexes}

The notion of fault complex was introduced in the study of the foliation of CSS codes~\cite{foliation,vasmer2025fault}. More specifically, the fault complex $\mathcal{F}$ is the product complex of a 3-term chain complex $\mathcal{Q}$ corresponding to CSS code and a length-$\ell$ repetition-code complex $\mathcal{R}_\ell$:
\begin{eqs}
     \mathcal{F}=\mathcal{Q} \times \mathcal{R}_\ell.
\end{eqs}
Since $\mathcal{F}$ is constructed as the tensor product of two complexes, its fault distance is given by $\min \{\ell,d\}$ where $d$ is the distance of $\mathcal{Q}$. This agrees with the standard intuition that a distance-$d$ code requires at least $d$ rounds of syndrome extraction to achieve fault-tolerant error correction. Otherwise, an undetectable fault of weight $\ell$, extending purely along the temporal direction of $\mathcal{R}_\ell$, will reduce the fault distance below the code distance.

More generally, the fault complex is defined by a 4-term chain complex $\mathcal{F}$: 
\begin{equation}
\label{eq:fault_complex}
\begin{tikzcd}[row sep=0.7em, column sep=0.7em]
F_3 \arrow[r, "\partial_3"] &
F_2 \arrow[r, "\partial_2"] &
F_1 \arrow[r, "\partial_1"] &
F_0\\
{\scriptstyle \text{dual detector}} &
{\scriptstyle \text{dual fault position}} &
{\scriptstyle \text{primal fault position}} &
{\scriptstyle \text{primal syndrome}}
\end{tikzcd}
\end{equation}
Here, the outer boundary maps define the primal and dual detector matrices $D_{\text{primal}}=\partial_1, D_{\text{dual}}^\mathsf T=\partial_3$. The intermediate boundary map $\partial_2$ 
specifies the connectivity between primal and dual fault positions and captures the equivalence relations between faults that produce the same detector syndromes.
Since $\mathcal{F}$ is a chain complex, the boundary operators satisfy $\partial_i \partial_{i+1}=0$. Concretely, \(\partial_1\partial_2=0\) ensures that primal fault patterns generated
by \(\partial_2\) have trivial primal syndrome, while
\(\partial_2\partial_3=0\) gives the corresponding consistency
condition on the dual side.

The associated fault distances are given by  
\begin{eqs}\label{eq:systolic_distance}
    d_{\text{primal}}&=d_1(\mathcal{F}_\bullet )=\min \{ |\mathbf{x}|:\mathbf{x}\in \mathrm{ker} (\partial_1) / \mathrm{im}(\partial_2) \},\\
    \quad d_{\text{dual}}&=d^2(\mathcal{F}^\bullet )= \min \{ |\mathbf{x}|:\mathbf{x}\in \mathrm{ker} (\partial_3^\mathsf T) / \mathrm{im}(\partial_2^\mathsf T) \} .
\end{eqs} 
These distances are the minimum weights of undetectable but homologically
nontrivial fault configurations, and therefore quantify the smallest number
of faults that can induce a nontrivial logical fault. In this framework, the primal and dual detectors correspond to the $X$- and $Z$-type parity checks, respectively. 
With the convention used here, primal fault positions track $Z$-type data errors
together with measurement errors in $X$-syndrome extraction, while dual fault
positions track $X$-type data errors together with measurement errors in
$Z$-syndrome extraction. Equivalently, the primal and dual detector matrices
give the parity-check constraints on the corresponding syndrome histories.
In the rest of the paper,  $\mathcal{F} $ denotes a fault complex.

A 4-term chain complex can also define valid CSS codes, with the correspondence between the code complex and fault complex interpretations summarized in Table~\ref{table:code_complex_fault_complex}. Unlike a code complex where the spaces $C_2$ and $C_1$ correspond to $Z$-checks and qubits, respectively, in a fault complex, the spaces $F_2$ and $F_1$ both represent qubits, namely dual and primal qubits. In circuit-based quantum error correction, the separation into primal and dual fault sectors reflects the decoupling of $X$- and $Z$-syndrome extraction in CSS codes: primal ($Z$-type) faults affect only $X$-syndromes, while dual ($X$-type) faults affect only $Z$-syndromes. In the measurement-based setting, this separation instead originates from the bipartite structure of the underlying cluster state. As a result, decoding can be carried out independently in the two sectors, a structure that is made explicit in the fault complex formalism.

A key distinction between code complexes and fault complexes lies in how their distance notions are encoded in homological data. The code distance is computed from (co)homology at the same location in the code complex, reflecting undetectable data errors. Specifically, the $Z$- and $X$-distances are the 1-systolic and 1-cosystolic distances corresponding to nontrivial cycles in the first homology and cohomology groups
$\frac{\mathrm{ker}(\partial_1)}{\mathrm{im}(\partial_2)}$, $\frac{\mathrm{ker}(\partial_2^\mathsf T)}{\mathrm{im}(\partial_1^\mathsf T)}$, respectively. In contrast, the fault distance arises from (co)homology groups at different (adjacent) levels of the fault complex, as seen from~\eqref{eq:systolic_distance}. A classic illustration is provided by the correspondence between the foliated 2D toric code and 3D toric code. Consider a $d\times d$ toric code tensor product with a length-$d$ repetition code. When the resulting chain complex is interpreted as the code complex corresponding to the standard 3D toric code, the $Z$- and $X$-code distances are
$d_Z=d,~d_X=d^2$. By contrast, when it is interpreted as a fault complex, the primal and dual fault distances are both $d_{\text{primal}}=d_{\text{dual}}=d$.
Hence, the required spacetime cost for the foliated surface code to achieve fault distance $d$ is $O(d^3)$. For an $N$-qubit stabilizer code with linear code distance $d\sim \Theta(N)$, we generally require a $d$-round syndrome extraction to achieve fault distance $d$. The resulting spacetime cost therefore scales as $O(Nd) \sim O(d^2)$. A more detailed explanation of this correspondence is in Appendix \ref{appendix:2Dto3D}.

\textcolor{black}{In this work, we show that every 4-term chain complex naturally defines a valid fault-tolerant protocol. }
Mathematically, \textcolor{black}{a large class of} 4-term chain complexes can be constructed from the product of a 3-term and a 2-term complexes. 
In the fault complex setting, such \textcolor{black}{product} constructions admit a natural interpretation as a \textcolor{black}{spacetime process}: the quantum code defined by the 3-term complex \textcolor{black}{propagates} along a temporal direction specified by the 2-term complex. \textcolor{black}{In contrast, a general 4-term chain complex need not possess an underlying product structure and may instead describe a more general quantum-information process on a graph, without a distinguished notion of temporal evolution.} 
From this viewpoint, the additional dimension encodes the spacetime evolution of the fault-tolerant process, while the structure of the 2-term complex enforces consistency between layers.
This perspective opens a pathway for importing rich techniques and analysis for quantum code construction into the design of fault complexes. We present a new concrete method for constructing low-overhead  fault complexes based on the lifted product method in Sec.~\ref{sec:spacetime_lifting}. \textcolor{black}{To our knowledge, this is the first family of fault complexes that simultaneously admits a clear operational interpretation and achieves favorable asymptotic fault-tolerance parameters.}

\begin{table*}[t]
\centering
\renewcommand{\arraystretch}{1.25}
 \begin{tabular}{c|c|c}
\hline
 & \textbf{Code complex} & \textbf{Fault complex} \\
\hline
Degree-3 component & $Z$-metacheck $C_3$ & Dual detector $F_3$\\
Degree-2 component & $Z$-check $C_2$ & Dual qubit $F_2$\\
Degree-1 component & qubit $C_1$ & Primal qubit $F_1$\\
Degree-0 component & $X$-syndrome $C_0$ & Primal syndrome $F_0$ \\
\hline 
$\partial_3$ & $M_Z^\mathsf T$& $D_\text{dual}^\mathsf T$ \\
$\partial_2$ & $H_Z^\mathsf T$&  Adjacency matrix connecting primal and dual qubit locations \\
$\partial_1$ & $H_X$&  $D_{\text{primal}}$  \\
\hline
Errors
& Static errors acting on data qubits $C_1$
& Spacetime fault processes acting on $F_2, F_1$ \\
\hline 
  Nontrivial undetectable errors
&
$\begin{aligned}
\mathcal{L}_Z &= 
\frac{\ker (\partial_1)}{\operatorname{im} (\partial_2)} \\
\mathcal{L}_X &= 
\frac{\ker (\partial_2^\mathsf T)}{\operatorname{im} (\partial_1^\mathsf T)}
\end{aligned}$
&
$\begin{aligned}
\mathcal{L}_{\text{primal}}& =

\frac{\ker (\partial_1)}{\operatorname{im} (\partial_2)}
\\
\mathcal{L}_{\text{dual}} &= 
\frac{\ker (\partial_3^\mathsf T)}{\operatorname{im} (\partial_2^\mathsf T)}
\end{aligned}$
\\
\hline
Spacetime meaning
& Static code
& Spacetime fault-tolerant protocol \\
\hline
Operational meaning
& Protection of logical quantum information
& \makecell{Protection of spacetime fault-tolerant processes \\ (e.g. teleportation, lattice surgery...)} \\

\hline
\end{tabular}
\caption{Comparison between the code complex and fault complex interpretations of a 4-term chain complex. The former characterizes protection against undetectable data errors, while the latter quantifies the robustness of fault-tolerant protocols against spacetime fault processes.}
\label{table:code_complex_fault_complex}
\end{table*}

\section{Measurement-based protocol from fault complex}

Fault complexes provide an algebraic description of spacetime fault-tolerant processes, but for general complexes their operational realization is not immediate. In this section, we show how every fault complex naturally gives rise to a measurement-based fault-tolerant protocol. 
In various previous works~\cite{newman2020generating,nickerson2018measurement,ftcomplex} the correspondence between 4-term cellular complexes and topological cluster states has been explored. Here, we generalize this connection by demonstrating how to construct a cluster state from an arbitrary fault complex, together with its operational interpretation.

Given a 4-term fault complex $\mathcal{F}$, we can construct a cluster state on a bipartite graph, where qubits are divided into primal and dual classes. Primal qubits correspond to elements of the vector space $F_1$, while dual qubits correspond to elements of $F_2$. The boundary map $\partial_2$ specifies the connectivity between primal and dual qubits. Note that edges in the graph only connect primal and dual qubits, with no primal-primal or dual-dual connections.

The cluster state stabilizer group is generated by
\begin{eqs}\label{eq:cluster_stabilizer}
\langle X_a Z_{\partial_2 a}, X_b  Z_{\partial_2^\mathsf T b} \rangle, \quad \forall a \in F_2, b \in F_1,
\end{eqs}
where $a$ and $b$ are binary vectors labeling dual and primal qubits, respectively. Here, the Pauli operators with subscripts $a,b,\partial_2 a, \partial_2^\mathsf T b$ denote products of $X$ and $Z$ operators supported on the qubits specified by the subscript. The $\partial_2 a$ denotes the collection of neighboring primal qubits around dual qubits $a$, while $\partial_2^\mathsf T b$ denotes the collection of neighboring dual qubits around primal qubits $b$.

Within the stabilizer group, certain operators are purely $X$-type. These pure $X$ operators include both primal/dual detectors as well as logical correlations. The primal and dual detectors are $\partial_1, \partial_3^\mathsf T$ respectively, and the chain complex conditions $\partial_1 \partial_2 = 0$ and $\partial_2 \partial_3 = 0$ ensure that each detector is a pure $X$ operator. 

Operationally, the fault complex describes a process as follows. One first prepares the cluster state defined by $F_2, F_1, \partial_2$. Then, all qubits are measured in the single-qubit $X$ basis. These destructive measurements eliminate all stabilizers that do not commute with the measurements, leaving only the primal and dual detectors as well as logical correlations. We emphasize that, in the MBQC setting, each detector corresponds to a constraint requiring a particular product of single-qubit $X$-measurement outcomes to be $+1$, rather than arising from non-destructive syndrome extraction as in the circuit model. In this setting, $Z$ errors that flip measurement outcomes are detected by these detectors. In contrast, $X$ errors commute with the $X$-basis measurements and thus do not produce nontrivial action on the measurement outcomes.

However, this description assumes that the entire system is measured, so that only classical correlations remain. For a genuinely quantum fault-tolerant protocol, quantum correlations must be preserved. Therefore, in practice, only a subset of qubits in the fault complex is measured.  This will be discussed later.

\section{Spacetime lifting}\label{sec:spacetime_lifting}

In Ref.~\cite{panteleev2021quantum}, Panteleev and Kalachev introduced the lifted product and obtained a family of qLDPC codes on \(N\) physical qubits that achieve parameters $k=\Theta(\log N)$ and $d=\Theta \left(\frac{N}{\log N}\right)$. They can be viewed as tensor products of classical parity-check matrices over the group algebra $R=\mathbb{F}_2[\mathbb{Z}_\ell]\cong \mathbb F_2[x]/(x^\ell-1)$, where each entry is a polynomial with binary coefficients subject to the relation $x^\ell = 1$, rather than a single binary value. For analytical convenience, we use this polynomial representation. In practice, each polynomial can be expanded into an $\ell \times \ell$ circulant binary matrix.

Given matrices $A\in R^{m_A\times n_A}$ and $B\in R^{m_B\times n_B}$,  the quasi-cyclic lifted product code $ \text{LP}(A,B)$ is defined as
\begin{eqs}\label{eq:LP_code_parity_check}
    \text{LP}(A,B): ~~ &H_X=\begin{pmatrix}
        A \otimes_{R} \mathbb{I}_{m_B}| \mathbb{I}_{m_A} \otimes_R B
    \end{pmatrix},\\
    &H_Z=\begin{pmatrix}
        \mathbb{I}_{n_A} \otimes_{R} B^\star| A^\star \otimes_R \mathbb{I}_{n_B}
    \end{pmatrix},
\end{eqs}
where ${}^\star$ denotes conjugate transpose over $R$, namely matrix
transpose together with the involution $x\mapsto x^{-1}$. When represented by binary form, the resulting code acts on $\ell (n_A m_B + m_A n_B)$ qubits, rather than the $\ell^2 (n_A m_B + m_A n_B)$ qubits that would arise from a tensor product. It is easy to check that the CSS condition is satisfied. Further background on lifted product codes is provided in Appendix~\ref {appendix:LP_properties}.

It is shown in Ref.~\cite{panteleev2021quantum} that $\text{LP}(A,1+x)$ can yield qLDPC codes with almost linear distance $d=\Theta\left(\frac{N}{\log N}\right)$ for a large lift size $\ell$. Here, $A$ is a $\mathbb{Z}_\ell$ quasi-cyclic classical expander code, while $1+x$ corresponds to a length-$\ell$ repetition code with periodic boundary conditions (denoted as $\mathcal{R}_\ell$). Conceptually, this construction can be viewed as a product of a classical expander complex and a repetition code complex, both carrying a $\mathbb{Z}_\ell$ cyclic symmetry, with the product taken over the common group algebra so as to quotient out a redundant cyclic
symmetry\footnote{The standard tensor product introduces two independent copies of the underlying cyclic symmetry, resulting in substantial redundancy and overhead in the product code. Consequently, the distance of an $N$-qubit tensor-product code is generally limited to at most $\Theta(\sqrt{N})$.} which can be written as a lifted product of chain complexes $\mathcal{A} \times_{\mathbb{Z}_\ell } \mathcal{R}_\ell$. This “twisted product” idea underlies much of the
modern progress on qLDPC constructions including lifted and balanced products \cite{panteleev2021degenerate,panteleev2021quantum,breuckmann2021balanced,panteleev2022asymptotically} and fiber bundle codes \cite{hastings2021fiber}, enabling major
improvements in qLDPC parameters leading up to the construction of asymptotically good qLDPC codes by taking products over group
rings.

\begin{figure*}[t]
    \centering
    \includegraphics[width=0.7\linewidth]{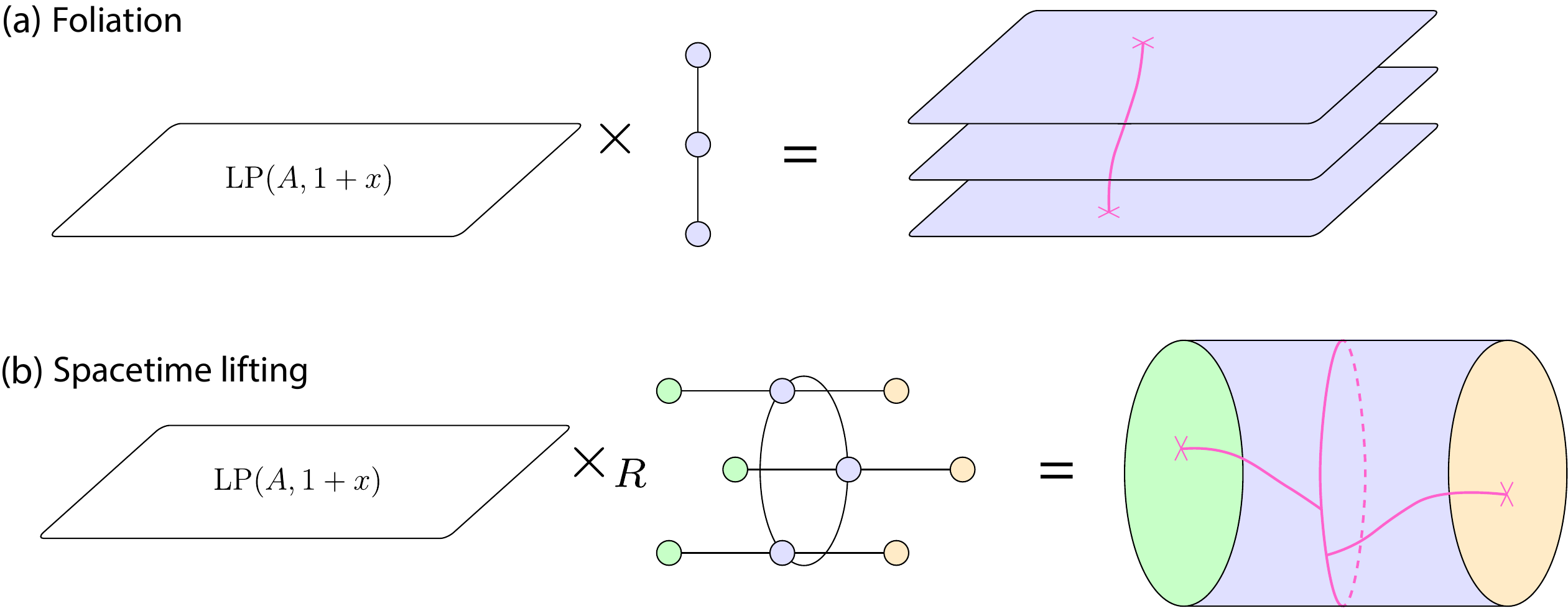}
    \caption{Comparison between (a)~foliation and (b)~spacetime lifting. In foliation, we take tensor product between a good lifted product code with a repetition code. The lifted product code has linear distance $d$, to achieve fault distance $d$, we have to let the repetition code to have at least length-$d$, so the total spacetime cost is \zw{$\widetilde O(d^2)$}. However, in spacetime lifting, we take a lifted product (tensor product quotiented out a $\mathbb{Z}_\ell$ cyclic symmetry) between a distance-$d$ lifted product code and a thickened length-$d$ repetition code with portal, to achieve fault distance being at least $d$, we only need \zw{$\widetilde O(d)$} spacetime cost instead of \zw{$\widetilde O(d^2)$}, which can also be interpreted as fault-tolerant syndrome extraction with $O(1)$ time overhead in a spacetime-lifted memory experiment. }
    \label{fig:comparison}
\end{figure*}

\subsection{Fault complexes from lifted product}

Building on the above idea, we extend beyond foliated constructions and construct fault complexes from lifted product complexes—a procedure we refer to as \emph{spacetime lifting}. Concretely, we construct 4-term complexes by iterating the lifted product over multiple codes: given three classical codes $A,B,C$, we define  $\text{LP}(A,B,C) \equiv \text{LP}(\text{LP}(A,B),C)$. 
This can be understood as a generalized foliation of the lifted product code $\text{LP}(A,B)$ with a temporal code $C$. Algebraically, it forms a 4-term chain complex through a lifted product over the group algebra, analogous to the usual tensor-product construction~\cite{armanda2021singleshot}. The three steps of chain maps are explicitly constructed from each classical code as  \begin{widetext}
\begin{eqs}\label{eq:LP_parity_check}
    \partial_1&= \left( \setlength{\arraycolsep}{6pt}\begin{array}{c c c}
        A \otimes_R \mathbb{I}_{m_B}\otimes_R \mathbb{I}_{m_C} & \mathbb{I}_{m_A}\otimes_R B \otimes_R \mathbb{I}_{m_C} & \mathbb{I}_{m_A} \otimes_R \mathbb{I}_{m_B} \otimes_R C
    \end{array} \right),\\
    \partial_2&=\left( \setlength{\arraycolsep}{6pt }\begin{array}{ccc}
\mathbb{I}_{n_A} \otimes_R B \otimes_R \mathbb{I}_{m_C} & \mathbb{I}_{n_A} \otimes_R  \mathbb{I}_{m_B} \otimes_R C & 0 \\
  A \otimes_R \mathbb{I}_{n_B} \otimes_R \mathbb{I}_{m_C}& 0 & \mathbb{I}_{m_A} \otimes_R  \mathbb{I}_{n_B} \otimes_R C \\
0 &   A \otimes_R\mathbb{I}_{m_B} \otimes_R \mathbb{I}_{n_C}&\mathbb{I}_{m_A} \otimes_R B \otimes_R \mathbb{I}_{n_C}
\end{array} \right),\\
\partial_3&= \left(  \setlength{\arraycolsep}{6pt} \begin{array}{c}
        \mathbb{I}_{n_A} \otimes_R \mathbb{I}_{n_B} \otimes_R C\\ 
        \mathbb{I}_{n_A} \otimes_R B \otimes_R \mathbb{I}_{n_C}\\
        A \otimes_R \mathbb{I}_{n_B} \otimes_R \mathbb{I}_{n_C}
    \end{array} \right),
\end{eqs}
where $A,B,C$ are $m_A\times n_A, m_B\times n_B, m_C\times n_C$ check matrices over $R$ for the classical codes specified by $\mathcal{A}, \mathcal{B}, \mathcal{C}$, respectively.

\end{widetext}


\textcolor{black}{As a simple example, we can construct a detector-completed spacetime-lifted fault complex $\widehat{\mathrm{LP}}(A,1+x,1+x)$ with almost-linear fault distance scaling $d=\Omega\big(N/\log N\big)$ from the lifted-product complex $\mathrm{LP}(A,1+x,1+x)$, which can be expressed as
\begin{eqs}
\mathcal{A} \times_{\mathbb{Z}_\ell} \mathcal{R}_\ell \times_{\mathbb{Z}_\ell} \mathcal{R}_\ell,
\end{eqs}
which naturally forms a 4-term chain complex of the form of \eqref{eq:fault_complex}.}

\textcolor{black}{The detector-completed spacetime-lifted fault complex $\widehat{\mathrm{LP}}(A,1+x,1+x)$ retains the same $\partial_1$ and $\partial_2$ as $\mathrm{LP}(A,1+x,1+x)$, but replaces $\partial_3$ with a detector-completed boundary map $\widetilde{\partial}_3$ that eliminates the additional short cocycles introduced during spacetime lifting.}


When the 4-term chain complex \textcolor{black}{$\widehat{\text{LP}}(A,1+x,1+x)$} is interpreted as a quantum code, the spaces $C_3, C_2, C_1, C_0$ correspond to $Z$-metachecks, $Z$-checks, $Z$-errors, and $X$-syndromes, respectively. The code distances are determined by the minimal nontrivial elements in the (co)homology groups $\ker(\partial_1)/\mathrm{im}(\partial_2)$ and $\ker(\partial_2^\mathsf T)/\mathrm{im}(\partial_1^\mathsf T)$.
Alternatively, when it is interpreted as a spacetime fault complex, the first and third boundary maps define the primal and dual detectors. In this case, the fault distances are governed by the (co)homology groups $\ker(\partial_1)/\mathrm{im}(\partial_2)$ and $\ker(\tilde{\partial}_3^\mathsf T)/\mathrm{im}(\partial_2^\mathsf T)$. The total spacetime cost of the fault complex is given by the total number of primal and dual qubits, namely $\dim(F_1)+\dim(F_2)$, which represents the total number of qubit locations in the spacetime process.
Informally, this corresponds to the spacetime volume---$\text{time} \times \text{space}$---in the circuit model. While implementation-specific optimizations such as qubit recycling may reduce the number of simultaneously active qubits, they do not change the overall spacetime-volume scaling.
Therefore, the spacetime cost provides a
model-independent measure of the overhead, which characterizes the required amount of physical resources. In this way, a code complex
with large \(Z\)-distance and large single-shot \(Z\)-distance naturally gives rise to a spacetime fault complex with large primal and dual fault distances.

To make the performance explicit, let
$A\in\mathbb{F}_2[\mathbb{Z}_\ell]^{m\times n\Delta}$ be
$\Delta$-limited, with $A$ and $A^\mathsf T$ both
$(\alpha,\beta)$-expanding and $A(1)$ of full row rank, where
\textcolor{black}{$\gamma=\frac{\alpha}{4\Delta}\min\{\beta,1\}$,
$n=\Theta(\log\ell)$}, and $m\leq n\Delta/2$.  Under this choice, the resulting fault-tolerant protocol on $O(\ell n\Delta )$ total physical qubits
 achieves an almost linear distance scaling $d\geq \gamma \ell = \Omega (N/\log N)$. This follows from the fact that the numbers of primal and dual qubits in the fault complex differ from those of the underlying lifted product code only by constant factors. 
 \zw{The distance bound holds for the detector-completed complex.
The completion changes only the dual detector map, leaving
$\partial_2$ and the number of cluster state qubits unchanged.}
 See Appendix~\ref{appendix:distance_proof} for the full proof with further details.

\subsection{Spacetime-lifted memory experiments with almost linear fault distance}

We now present a fault complex for the spacetime-lifted memory experiment of the lifted product code $\mathrm{LP}(A,1+x)$, which achieves almost-linear fault distance. In contrast, implementing the conventional memory experiment on lifted product code $\mathrm{LP}(A,1+x)$ using a standard foliated fault complex achieves at best square-root fault distance scaling with the total spacetime cost.

To model the spacetime-lifted memory experiment, it is necessary to specify the initial and final boundary conditions. To this end, we consider the \textcolor{black}{detector-completed} spacetime-lifted fault complex $\widehat{\text{LP}}(A,1+x,C)$, where $C$ is a “thickened” repetition code (see Fig.~\ref{fig:comparison}(b)) given by
\begin{eqs}
C=\begin{pmatrix}
        1+x & 0& 0\\
        1& 1 & 0\\
        1& 0& 1
    \end{pmatrix}.
    \end{eqs} 
    This code, illustrated in Fig.~\ref{fig:repetition_portal}(b), has a trilayer structure that encodes the start-middle-end causal structure of the spacetime-lifted memory experiment. The resulting construction can be interpreted as a spacetime-lifted realization of a three-round measurement protocol with nontrivial dual correlation. Since all the qubits are destructively measured in such process, this protocol corresponds to a spacetime-lifted memory experiment.
    
    Using the expansion properties of $A$ and $A^\mathsf T$, we verify  that the resulting fault-tolerant protocol has total spacetime cost $N=O(\ell n \Delta)$ and the fault distance of this memory experiment satisfies $d \geq \gamma \ell = \Omega\left(\frac{N}{\log N}\right)$. The full derivation is detailed in Appendix~\ref{appendix:memory_distance}. For comparison, for the same underlying quantum code $\text{LP}(A,1+x)$, a conventional three-round foliated construction yields a fault complex with distance only $3$.

This example highlights the fundamental advantage of spacetime lifting over standard foliation: by leveraging the symmetries of the underlying code into the design of the fault-tolerant protocol, one can improve the fault distance scaling without increasing the asymptotic overhead.

\section{Fault complexes for logical teleportation}

As discussed above, a fault complex naturally describes a fault-tolerant process in which all qubits are ultimately measured, leaving only classical correlations. This fully measured process may be viewed as the bulk component of a fault-tolerant protocol. To process and transmit quantum information, however, a subset of qubits must remain unmeasured and serve as input and output portals of a fault-tolerant quantum channel. In this section, we study the conditions under which a fault complex supports logical teleportation, one of the central primitives of fault-tolerant quantum computation.
Previous works \cite{newman2020generating, nickerson2018measurement, ftcomplex} studied measurement-based logical teleportation in the setting of 3D cellular complexes and topological codes, where the input/output portals and correlations between them arise naturally from geometry. Operationally, these protocols can be interpreted as fault-tolerant measurement-based preparation of a logical Bell pair shared between input and output. By the Choi–Jamiołkowski isomorphism, such state serves as a stabilizer resource state implementing the logical identity channel \cite{bombin2023logical, beverland2024fault}. Consequently, consuming this resource state together with fault-tolerant state injection at the input portal—implemented, for example, via transversal Bell measurements between a noisy encoded state and the input portal, realizes fault-tolerant logical teleportation. However, a general framework for constructing fault complexes and fault-tolerant protocols that enable logical teleportation remains lacking. In this section, we address this gap by formulating a class of measurement-based fault-tolerant protocols for logical teleportation of CSS codes from general fault complexes.

Our goal is to have two portals consisting of two code blocks entangled by the logical correlations $\overline{X_{\text{in}} X_{\text{out}}}$ and $\overline{Z_{\text{in}} Z_{\text{out}}}$ \footnote{Generally, the input and output portals encode $k$ logical qubits. For notational simplicity, we write $\overline{X_{\text{in}} X_{\text{out}}}$, $\overline{Z_{\text{in}} Z_{\text{out}}}$ to denote the complete collection of logical Bell correlations for all $k$ logical Bell pairs.}. Equivalently, the portals should be stabilized by the logical Bell-state stabilizer group $\langle \mathcal{S}_\text{in}, \mathcal{S}_\text{out}, \overline{X_{\text{in}} X_{\text{out}}}, \overline{Z_{\text{in}} Z_{\text{out}}} \rangle$\footnote{For simplicity, we suppress the $\pm 1$ phases in the stabilizer generators which depend on the measurement outcomes, as the $\pm 1 $ signs can be removed by updating Pauli frame in post-processing.}. {To realize this, we consider a general fault complex $\mathcal{F}$ with two designated subsets of primal qubits, $Q_{\mathrm{in}},Q_{\mathrm{out}}\subset F_1$, called portals, which serve as the input and output code blocks.} We prepare the cluster state defined by $\mathcal{F}$ and perform single-qubit $X$ measurements on all qubits except those in $Q_{\text{in}}$ and $Q_{\text{out}}$. The resulting state on $Q_{\text{in}}$ and $Q_{\text{out}}$ lies in a logical Bell basis, with the specific eigenvalues determined by the bulk measurement outcomes.

A fault complex for logical teleportation should contain the following ingredients: 
\begin{enumerate}
    \item  \textbf{Fault tolerance}: Extensive masked primal fault distance and unmasked dual fault distance.
    \item \textbf{Logical encoding}: The input and output portals are encoded into the logical spaces of the stabilizer groups $\mathcal{S}_{\text{in}}$ and $\mathcal{S}_{\text{out}}$, respectively, where the two stabilizer groups are equivalent up to exchanging $Q_{\text{in}}$ and $Q_{\text{out}}$.
    \item  \textbf{Logical correlation}: There exist logical correlations $\overline{X_{\text{in}} X_{\text{out}}}$ and $\overline{Z_{\text{in}} Z_{\text{out}}}$.
\end{enumerate}

 For the fault tolerance condition, after measuring out qubits on $F_1 \backslash (Q_{\text{in}} \bigcup Q_{\text{out}})$ and $F_2$, we have access to all the dual detectors $\partial_3^\mathsf T=D_{\text{dual}}$ and masked primal detectors $\partial_{1,\text{mask}}=D_{\text{primal} \backslash Q_{\text{in}}, Q_{\text{out}}}:= \{ D \in D_{\text{primal}}: \text{supp}(D) \bigcap (Q_{\text{in}}\bigcup Q_{\text{out}}) =\emptyset \}$ since qubits in $Q_{\text{in}}$ and $Q_{\text{out}}$ are not measured, and therefore any detector with support on these regions is inaccessible. 
 The qubits on blocks $Q_{\text{in}}, Q_{\text{out}}$ lie in the eigenspace of some stabilizer group, as well as in the eigenspaces of the logical correlations. 
 The eigenvalues of logical correlations are inferred from the bulk measurement outcomes. Because these outcomes are noisy, error correction must be performed using the available detector information from $D_{\text{dual}}$ in $F_2$ and $D_{\text{primal} \backslash (Q_{\text{in}} \bigcup Q_{\text{out}})}$ in $F_1 \backslash (Q_{\text{in}} \bigcup Q_{\text{out}})$. Importantly, this correction is purely classical: it consists of inferring a consistent measurement outcome pattern without any quantum feedback. Hence the primal fault is detected by masked primal detectors $\partial_{1, \text{mask}}$ to correct primal faults on a reduced set of qubits $F_1 \backslash (Q_{\text{in}} \bigcup Q_{\text{out}})$. Accordingly, the primal fault distance is defined as $d_{\text{primal}} = \min \{ |\mathbf{x}|: \mathbf{x} \in \mathrm{ker}(\partial_{1, \text{mask}})/ \mathrm{im} (\partial_2)\}$, as introduced in the context of spacetime codes \cite{fu2025error}. The dual fault distance $d_{\text{dual}}$ remains unchanged from its definition in~\eqref{eq:systolic_distance}.

The logical encoding condition highlights the importance of the product structure: fault complexes of product type provide natural input/output portals for which the condition is automatically satisfied. Suppose the base quantum code given by $\mathcal{A} \times_R \mathcal{B}$ has check matrices $H_X, H_Z$, which gives stabilizers $\mathcal{S}=\langle S_{X,1},\ldots, S_{X,r_X}, S_{Z,1},\ldots, S_{Z,r_Z}\rangle  $ for a lifted product code specified by $\mathcal{A} \times_{R} \mathcal{B}$. The product structure naturally introduces multiple copies of the base code supported on different slices associated with the temporal complex $\mathcal{C}$. In particular, the first row of blocks of $\partial_2$ in \eqref{eq:LP_parity_check} contains $m_C$ copies of the $H_Z$ checks acting on different slices. Choosing two such slices as the portals $Q_{\text{in}},Q_{\text{out}}$ therefore guarantees that the two portals support the same encoded stabilizer code. More precisely, for every $i=1,\ldots,r_Z$, there exist dual qubits $b_i,b_i'$ satisfying $\partial_2 b_i=\text{supp}(S_{Z,i, \text{in}}), \partial_2 b_i'=\text{supp}(S_{Z,i,\text{out}})$, so that the associated cluster state stabilizers take the form $X_{b_i} Z_{\partial_2 b_i}=X_b S_{Z,i,\text{in}}$, $X_{b_i'} Z_{\partial_2 b_i'}=X_b S_{Z,i,\text{out}}$. Hence, the $Z$-stabilizers of the code on the portals are naturally embedded into the cluster state stabilizer group on each portal. Similarly, from the structure of $\partial_1$ in \eqref{eq:LP_parity_check}, there are $m_C$ copies of the $H_X$ checks coupled with temporal qubits from $\mathcal{C}$ supported on different slices. Using the same choice of portals, one obtains primal detectors of the form  $D_{X,i}= S_{X,i,\text{in}} X_{\text{bulk}}\equiv 1, D_{X,i}'= S_{X,i,\text{out}} X_{\text{bulk}}\equiv 1$, for all $i=1,\ldots,r_X$, where $X_{\text{bulk}}$ denotes the product of bulk $X$-measurement outcomes appearing in the detector. Consequently, the eigenvalues of the $X$-stabilizers on both portals can also be inferred directly from the bulk measurement outcomes.

For the logical correlation condition, the two portals must also lie in an eigenbasis of the logical correlations $\overline{X_{\text{in}}X_{\text{out}}}$ and $\overline{Z_{\text{in}}Z_{\text{out}}}$, where $\overline{X}_{\text{in}},\overline{Z}_{\text{in}}$ denote the logical Pauli operators on the input portal, and similarly for the output portal. Since the portals are supported on primal qubits, a valid construction requires the existence of cluster state stabilizers that can be written as $\overline{X_{\text{in}} X_{\text{out}}} \cdot X_{q}$ with $q\in F_1 \backslash (Q_{\text{in}}\bigcup Q_{\text{out}})$. Such operators are characterized by the cohomology group $\ker(\partial_2^\mathsf T)/\mathrm{im}(\partial_1^\mathsf T)$. After measuring the bulk and performing classical decoding, the corrected bulk outcomes determine the eigenvalue of $\overline{X_{\text{in}} X_{\text{out}}}$.

Similarly, having logical correlation $\overline{Z_{\text{in}} Z_{\text{out}}}$ requires there exist a set of dual qubits $p \in F_2$ such that $\partial_2 p= \text{supp}(\overline{Z_{\text{in}}Z_{\text{out}}})$. The corresponding cluster state stabilizer $X_p Z_{\partial_2 p} = X_p \overline{Z_{\text{in}} Z_{\text{out}}}$ allows one to infer the eigenvalue of the logical $Z$ correlation directly from the measurement outcomes on $p$. It is important to distinguish this logical $\overline{Z_{\text{in}}}\overline{Z_{\text{out}}}$ correlation from the dual correlations in $\ker (\partial_2) / \mathrm{im}(\partial_3)$ of the fault complex. The $\overline{Z_{\text{in}}Z_{\text{out}}}$ correlation we look for should be a cluster state stabilizer which only contains $X$ in the bulk and logical $\overline{Z}$ at the two portals; however, the dual correlation of fault complex are pure $X$-type operators in the cluster state stabilizer acting solely on dual qubits. We use a 1D cluster state example to demonstrate the distinction between dual correlation and $\overline{Z_{\text{in}}Z_{\text{out}}}$ correlation in Appendix.~\ref{appendix:1D_cluster}.

This analysis gives an algebraic generalization of the Raussendorf–Bravyi–Harrington protocol \cite{RBH,raussendorf2007topological} and foliated codes \cite{foliation}, {promoting the previous teleportation picture based on foliation fault complexes} to a homological teleportation criterion for general fault complexes with
product structure. 
A key intuition worth highlighting is that the simple temporal structure in standard foliation may be rather inefficient, so our approach based on more general product complexes enables new opportunities to reduce the spacetime overhead.

\section{Conclusion and outlook}

In this work, we establish a systematic connection between fault complexes—an algebraic framework for describing spacetime fault-tolerant protocols—and static code complexes. Conceptually, this amounts to promoting a quantum code described by a 3-term chain complex into a fault-tolerant protocol described by a 4-term fault complex through the introduction of an additional spacetime dimension.   Exploiting this connection, we formalize the spacetime lifted product method, which utilizes symmetry reduction to construct fault complexes with substantially improved fault tolerance parameter behaviors compared to known methods. For example, we present a spacetime-lifted memory experiment whose fault distance scales almost linearly with the total spacetime cost.

Crucially, spacetime lifting not only achieves nearly optimal fault distance, breaking the square-root barrier of standard foliation, but also retains a natural operational interpretation in terms of fault-tolerant quantum information processing. To make this rigorous, we identify general conditions under which a fault complex realizes logical teleportation. While existing teleportation protocols arise naturally from topological or foliated product constructions, a general fault complex with good parameters may not support logical teleportation. We show that the product structure are naturally well suited for low-overhead logical teleportation protocols, providing a key operational motivation for studying them. Together, these results provide a framework for designing fault complexes that combine favorable parameter scaling with clear operational meaning.

The spacetime lifting perspective also brings several conceptual advances. First, it treats the entire spacetime process itself as the fundamental object, thereby incorporating spacetime overhead into the analysis of fault tolerance. Second, it opens a new route toward low-overhead fault-tolerant protocols drawing on techniques from the theory of qLDPC codes. In this sense, our approach is complementary to spacetime codes~\cite{bacon2017sparse,gottesman2022opportunities,delfosse2023spacetime,pesah2025fault}: while spacetime codes are constructed from fault-tolerant circuits, here we instead construct fault-tolerant protocols from higher-dimensional complexes which can define quantum codes.  Recent work has also suggested that fault tolerance plays a key role in efforts woward the quantum PCP conjecture
~\cite{anshu2024circuit,breuckmann2026adversarial}. Low-overhead fault complexes may offer a complementary framework for pursuing
this connection.

Although we focus in this work on lifted product constructions with cyclic symmetry, it is natural to expect that the framework extends further. In particular, given that lifted and balanced products with richer (e.g., non-Abelian) symmetry structures have led to major advances in qLDPC code parameters~\cite{panteleev2022asymptotically,breuckmann2021balanced}, it would be interesting to explore fault complexes and associated protocols leveraging such more general symmetry reduction. Furthermore, while this work investigates measurement-based protocols, a natural next step is to study circuit-based fault-tolerant protocols from fault complexes.

More broadly, fault complexes can also describe spacetime processes with nontrivial logical action such as lattice surgery \cite{cowtan2025fast}. From this viewpoint, logical teleportation corresponds to the case of transporting a logical identity channel, while more general fault-tolerant protocols can be interpreted as stabilizer channels implementing nontrivial logical operations where the input and output portals exhibit correlations beyond the Bell-type. Therefore, it would be interesting to explore how these ideas can be extended to construct other fault-tolerant primitives, which would serve as building blocks for low-overhead universal fault-tolerant quantum computation.

\section*{Acknowledgement }
YX thanks Yujie Zhang and Yilun Li for hosting his visits in Tokyo.
YW is supported by startup funding from SIMIS. 
ZWL is supported in part by NSFC under Grant No.~12475023, Dushi Program, and a startup funding from YMSC.

\bibliography{biblo.bib}

\onecolumngrid
\appendix

\section{3D toric code from 2D toric code foliation}\label{appendix:2Dto3D}

\begin{figure}[h]
    \centering
    \includegraphics[width=0.5\textwidth]{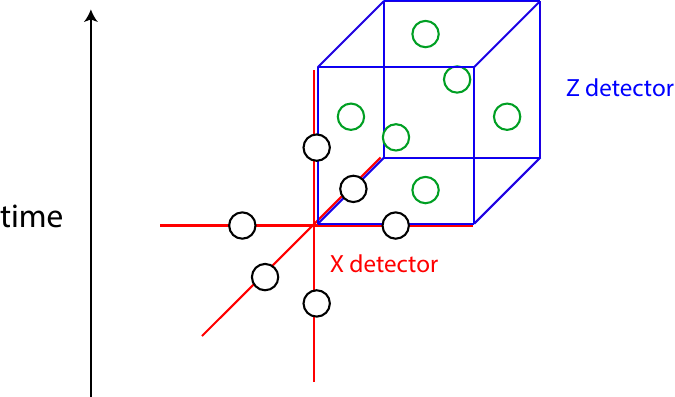}
    \caption{Foliated 2D toric code, where the black cycles represent primal faults associated with a $X$ (primal) detector, and the green cycles represent dual faults associated with a $Z$ (dual) detector. Here we consider the phenomenological error model.}
    \label{fig:foliated_TC}
\end{figure}

We illustrate this structure using the correspondence between the fault complex of the foliated 2D toric code and the 3D toric code complex. It is well known that taking the tensor product of a 2D toric code with a repetition code produces the chain complex of the 3D toric code. Equivalently, the 3D toric code can be interpreted as a stack of 2D toric code layers coupled by a repetition code along the third (temporal) direction. From this perspective, the repetition code supplies the inter-plane constraints that glue together successive 2D toric code sheets. Here, we make explicit the correspondence between the code complex of the 3D toric code and the fault complex arising from the foliated 2D toric code.

In the standard 3D toric code, physical qubits are placed on edges $C_1$. The $X$-checks are weight-6 operators supported on vertices $C_0$, while the $Z$-checks are weight-4 operators supported on faces $C_2$. The $Z$-metachecks are products of $Z$-checks surrounding closed volumes in $C_3$, thereby enabling single-shot correction of $X$ errors.

This layered viewpoint aligns naturally with the fault complex interpretation: each 2D layer corresponds to a round of syndrome extraction, while the inter-layer couplings encode measurement consistency between adjacent rounds. As a result, spacetime faults in the foliated 2D toric code are promoted to geometric objects in the 3D cellular complex.

More concretely, this chain complex can be reinterpreted as a fault complex as shown in Fig.~\ref{fig:foliated_TC}. In this picture, primal fault positions reside on edges, while dual fault positions reside on faces. The primal fault space captures $Z$-type data errors of the foliated 2D toric code together with measurement errors in $X$-check readout. Correspondingly, the dual fault space captures $X$-type data errors together with measurement errors in $Z$-check readout.

Under this correspondence, $Z$-checks of the 3D toric code are identified with dual fault positions in the foliated complex. The primal detectors are weight-6 spacetime checks that combine a weight-4 $X$-check of the 2D toric code in a time slice with measurement errors from two adjacent rounds. Similarly, the dual detectors are weight-6 checks acting on dual fault positions, corresponding to the $Z$-metachecks that combine a weight-4 $Z$-check with measurement errors across neighboring rounds.

In this sense, the 3D toric code provides a concrete illustration of this correspondence, where the complex encodes the fault-tolerant spacetime structure of the foliated 2D toric code: cube constraints enforce detector consistency across adjacent rounds, and the third dimension directly represents the temporal direction of repeated syndrome extraction.

\section{Basic properties of lifted product codes}
\label{appendix:LP_properties}

We denote the binary representation of a matrix $M \in \mathbb{F}_2[\mathbb{Z}_\ell]^{m\times n}$ by $\mathbb{B}(M)\in \mathbb{F}_2^{\ell m \times \ell n}$. For a matrix with entries in $\mathbb{F}_2[\mathbb{Z}_\ell]$, the binary representation is obtained by replacing each polynomial entry with the corresponding sum of cyclic shift matrices. In particular, the generator $x$ is represented by the $\ell \times \ell$ cyclic translation matrix whose generator, inverse, and identity are
\begin{eqs}
    \mathbb{B}(x)= \begin{pmatrix}
        0 & 1& 0& \cdots & 0\\
        0 & 0& 1& \cdots & 0\\
        \vdots & \vdots & \vdots & \ddots & \vdots\\
        1 & 0& 0& \cdots &0
    \end{pmatrix},\quad \mathbb{B}(x^{-1})= \mathbb{B}(x)^\mathsf T
\end{eqs}
and $\mathbb{B}(1)= \mathbb{B}(x) \mathbb{B}(x^{-1})=\mathbb{I}_{\ell}$. 
Given the binary representation of the polynomial ring $\mathbb{F}_2[\mathbb{Z}_\ell]$, we can immediately show that the binary representation of a $1\times 1$ matrix $1+x$ is the parity check matrix of length-$\ell$ repetition code with periodic boundary condition
\begin{eqs}
    \mathbb{B}(1+x)=\mathbb{B}(1)+\mathbb{B}(x)= \begin{pmatrix}
        1 & 1& 0& \cdots & 0& 0\\
        0 & 1& 1& \cdots & 0& 0\\
        \vdots & \vdots & \vdots &\ddots & \vdots & \vdots\\
        1 & 0& 0& \cdots & 0& 1
    \end{pmatrix}.
\end{eqs}
The Hamming weight $|\cdot |$ extends naturally to the polynomial ring $\mathbb{F}_2[\mathbb{Z}_\ell]$, by defining it as the number of nonzero monomial terms; for example, $|\mathbb{B}(x+x^2)|=|x + x^2| =2$. Hence the Hamming weight of a matrix over the polynomial ring is defined as the Hamming weight of its binary representation.

Aside from the Hamming weight of matrices, to analyze the weight of codeword vectors written in terms of vectors over polynomial ring, the binary representation $\mathbb{b}$ of a vector $v=(v_1,\ldots,v_n)\in \mathbb{F}_2[\mathbb{Z}_\ell]^n$ where every entry are written as $v_i=\sum_{j=0}^{\ell-1}v_{i,j} x^j$ is defined as 
\begin{eqs}
    \mathbb{b}(v)=\begin{pmatrix}
        v_{1,0}, v_{1,1}, \ldots, v_{n,\ell-1}
    \end{pmatrix}\in \mathbb{F}_2^{\ell n},
\end{eqs}
which is a collection of all the $\mathbb{F}_2$ coefficients in $v$.

Again, the Hamming weight of $v$ is equal to the Hamming weight of the corresponding binary representation $\mathbb{B}(v)$. The binary representation has following identity $\mathbb{b}(Mv)= \mathbb{B}(M) \mathbb{b}(v)$. 

The lifted product over $\mathbb{F}_2 [\mathbb{Z}_\ell]$ implements the identification
\begin{eqs}
     x^{h_1}\otimes_R x^{h_2+1} \equiv  x^{h_1+1}\otimes_R x^{h_2}. 
\end{eqs}
In the lifted product construction, we first take lifted product between chain complexes over $\mathbb{F}_2[\mathbb{Z}_\ell]$ and then pass to a binary representation to obtain a complex over $\mathbb{F}_2$.

\section{Proof of the main theorems}\label{appendix:distance_proof}

In this appendix, we prove that the \textcolor{black}{spacetime-lifted} fault complex $\text{LP}(A,1+x,1+x)$ achieves almost linear scaling of both 1-systolic (primal fault) and 2-cosystolic (dual fault) distances. Our proof strategy follows that of Ref.~\cite[Sec. V]{panteleev2021degenerate}, which establishes almost linear code distance for $\text{LP}(A,1+x)$ by analyzing the corresponding $1$-systolic ($X$) and $1$-cosystolic ($Z$) distances of the associated complex.

The key differences in our setting are twofold. First, we consider a complex obtained from the lifted product of three classical codes, rather than two. Second, we aim to establish almost linear scaling for the $1$-systolic and $2$-cosystolic distances, instead of the $1$-systolic and $1$-cosystolic distances considered in the analysis of lifted product complex $\text{LP}(A,1+x)$ in Ref.~\cite{panteleev2021quantum}. \textcolor{black}{Moreover, the raw iterated lift contains an inherited primal
sector and an ancillary temporal sector corresponding to
$\operatorname{coker}A(1)$. The full-row-rank condition on
$A(1)$ makes the latter trivial, which is why every nontrivial
primal class satisfies $u(1)\neq0$. The analogous ancillary
sector on the dual side need not vanish and is removed by the
detector completion introduced in the proof of Theorem~\ref{thm:linear_distance}. }

\begin{definition}[Expanding matrix]
    For a binary $H \in \mathbb{F}_2^{m\times n}$, we call $H$ is $(\alpha,\beta)$-expanding if for all $|y|\leq \alpha n $ and $y\in \mathbb{F}_2^{ n}$ we always have $|H y| \geq \beta|y| $. Here $\alpha,\beta$ are real number and $\alpha <1$.

    For a matrix $A \in \mathbb{F}_2[\mathbb{Z}_\ell]^{m\times n}$ over the polynomial ring, we call $A$ is a $(\alpha,\beta)$-expanding matrix if its binary representation $\mathbb{B}(A) \in \mathbb{F}_2^{\ell m \times \ell n}$ is $(\alpha,\beta)$-expanding such that for all $|y|\leq \alpha \ell n $ and $y\in \mathbb{F}_2^{ \ell n}$ we always have $|H y| \geq \beta|y| $.
\end{definition}

\begin{theorem}[Almost-linear \zw{primal and dual fault distances}]\label{thm:linear_distance}
    Let $A\in \mathbb{F}_2 [\mathbb{Z}_\ell]^{m\times n\Delta}$ be $\Delta$-limited, i.e., every row and column of its binary representation $\mathbb B(A)$ has Hamming weight at most $\Delta$. If $A$ and $A^\mathsf T$ are $(\alpha,\beta)$-expanding \textcolor{black}{and $A(1)$ has full row rank}, then the 1-systolic and 2-cosystolic distances of the fault complex $\widehat{\text{LP}}(A,1+x,1+x)$ are at least $\gamma \ell$, where \zw{$\gamma=\frac{\alpha}{4\Delta}\min\{\beta,1\}$.}
    \zw{Here $\widehat{\text{LP}}(\cdot)$ denotes the dual detector completion of $\text{LP}(\cdot)$, obtained by setting $\widetilde{\partial}_3=(\partial_3\;Y)$, where $Y=(y_1,\ldots,y_r)$, $y_i=(j_\ell\mathbf{t}_i,0,0)^\mathsf T$, and $r=\dim\ker A(1)=n\Delta-m$. The vectors $\mathbf{t}_1,\ldots,\mathbf{t}_r$ form a basis of $\ker A(1)$. This modification leaves the cluster state connectivity $\partial_2$ and the primal detector map $\partial_1$ unchanged.}
\end{theorem}

\begin{proof}
    Because $A$ is $(\alpha,\beta)$-expanding, its kernel distance satisfies 
$d(A)\equiv
\min_{\mathbf{x}\in\ker A\backslash\{\mathbf 0\}}
|\mathbf{x}|
\geq \alpha\ell n\Delta$. 
    For a codeword $c \in \mathrm{ker} (A(1))$ such that $A(1)c=0$, we have 
    \begin{eqs}
        j_\ell A(1)c= A j_\ell c=\mathbf{0},
    \end{eqs}
    so $j_\ell c \in \mathrm{ker}(A)$. Consequently, 
    \begin{eqs}
     \alpha \ell n \Delta \leq d(A)\leq  \ell d(A(1)) \quad \Rightarrow \quad d(A(1))\geq \alpha n \Delta.
    \end{eqs}

    Let $(u,v,w)$ be a representative of a nontrivial class in $\ker\partial_1/\operatorname{im}\partial_2$. We first consider the case $u(1)\neq\mathbf{0}$.

    If $|u| \geq \frac{\alpha \ell n \Delta}{2}$, the desired bound follows. 
    Otherwise, we show that if $|u| \leq \alpha \ell n \Delta/2$, then $|v|+|w|$ is at least $\frac{\alpha \beta \ell n \Delta}{4}$. We define 
    \begin{eqs}
        u^{(t)}=j_t u, s^{(t)}=Au^{(t)},
    \end{eqs}
    where $j_t=\sum_{k=0}^{t-1}x^k$. 
    Using the kernel condition $Au+(1+x)(v+w) =\mathbf{0}$,
    \begin{eqs}
        s^{(t)}=Au j_t= (1+x)(v+w) j_t= (1+x^t) (v+w).
    \end{eqs}
    Because $|u^{(\ell)}|= |j_\ell u(1)| \geq \alpha \ell n \Delta$, there exists a minimal $t_0$ such that $|u^{(t_0+1)}| \geq \alpha \ell n \Delta$. Then
    \begin{eqs}
        u^{(t_0)}= u^{(t_0+1)}+ x^{t_0}u,
    \end{eqs}
    so
    \begin{eqs}
        |u^{(t_0)}|= |u^{(t_0+1)}+ x^{t_0}u| \geq |u^{(t_0+1)}|- |u| \geq \frac{\alpha \ell \Delta n}{2}.
    \end{eqs}
    Since $s^{(t_0)}=(1+x^{t_0})(v+w)$, we have $|v+w| \geq \frac{1}{2} |s^{(t_0)}|$. Because $A$ is $(\alpha,\beta)$-expanding and $ \alpha \ell  n \Delta/2 \leq |u^{(t_0)}| \leq \alpha \ell n \Delta$, we obtain $|s^{(t_0)}|=|A u^{(t_0)}|\geq \beta |u^{(t_0)}|$.
    Therefore
    \begin{eqs}
        |v|+|w| \geq |v+w|  \geq \frac{1}{2} |s^{(t_0)}|=\frac{1}{2}|A u^{(t_0)}|\geq \frac{\beta}{2} |u^{(t_0)}|\geq  \frac{\alpha \beta \ell n \Delta}{4}.
    \end{eqs}
    Combining the two subcases, we have 
    \zw{\begin{eqs}
    |(u,v,w)|\geq \gamma_1 \ell n\geq\gamma_1 \ell,\qquad\gamma_1=\min\left\{\frac{\alpha \Delta }{2},
    \frac{\alpha\beta \Delta}{4}\right\}.
    \end{eqs}}
    \textcolor{black}{which gives the lower bound of primal fault distance.}

\textcolor{black}{It remains to show that every nontrivial class satisfies $u(1)\neq 0$. Suppose that $u(1)=0$, so $u=(1+x)h$ for some $h\in R^{n\Delta}$, and set $p=v+Ah$. Adding the boundary $\partial_2(h,0,0)$ gives $(u,v,w)\sim(0,p,w)$. The cycle condition implies $p+w=j_\ell q_0$ for some $q_0\in\mathbb{F}_2^m$. Since $A(1)$ is surjective, choose $k\in\mathbb{F}_2^{n\Delta}$ such that $A(1)k=q_0$. Adding $\partial_2(j_\ell k,0,0)=(0,j_\ell q_0,0)$ then reduces $(0,p,w)$ to $(0,w,w)$. Next choose $k'\in\mathbb{F}_2^{n\Delta}$ such that $A(1)k'=w(1)$. Since $(w+Ak')(1)=0$, there exists $r_0\in R^m$ such that $w+Ak'=(1+x)r_0$. Therefore,
\begin{eqs}
    (u,v,w)\sim(0,p,w)\sim(0,w,w)=\partial_2(k',k',r_0).
\end{eqs}
Hence the original class is trivial. Consequently, every nontrivial undetectable primal error satisfies $u(1)\neq 0$.}

\zw{
Since $A(1)$ has full row rank, the raw dual sector contains
$r=n\Delta-m$ short classes with weight-$2$ representatives
$\eta_j=(\mathbf e_j,\mathbf e_j,0)\in F_2$, where the
$\mathbf e_j$ are standard basis vectors whose cosets form a
basis of $\operatorname{coker}A(1)^\mathsf T$. Choose
$\mathbf t_1,\ldots,\mathbf t_r$ as the dual basis of
$\ker A(1)$ under the natural pairing. To eliminate these short cocycles, we add the dual detectors
$y_i=(j_\ell\mathbf t_i,0,0)^\mathsf T$, where
$\mathbf t_1,\ldots,\mathbf t_r$ form a basis of
$\ker A(1)$ and $r=n\Delta-m$. Writing
$K=(\mathbf t_1,\ldots,\mathbf t_r)
\in\mathbb F_2^{n\Delta\times r}$
and $Y=(y_1,\ldots,y_r)$, we define the completed dual detector matrix as
\begin{eqs}
\widetilde{\partial}_3^\mathsf T=
\begin{pmatrix}
\partial_3^\mathsf T\\
Y^\mathsf T
\end{pmatrix}.
\end{eqs}
These are valid detectors because
$\partial_2y_i=(0,j_\ell A(1)\mathbf t_i,0)^\mathsf T=0$.
Choose representatives $\mathbf e_1,\ldots,\mathbf e_r$ whose
cosets form a basis of $\operatorname{coker}A(1)^\mathsf T$, and choose
$\mathbf t_1,\ldots,\mathbf t_r$ as the dual basis in
$\ker A(1)$ under the natural pairing. Then, for the short
cocycles $\eta_j=(\mathbf e_j,\mathbf e_j,0)$,
\begin{eqs}
    \langle y_i,\eta_j\rangle_{\mathbb F_2}
=\mathbf t_i^\mathsf T\mathbf e_j=\delta_{ij}.
\end{eqs}
Consequently, the added dual detectors detect the entire
short-cocycle sector.}

\zw{This procedure can be viewed as gauging out the additional
logical degrees of freedom without introducing extra quantum
overhead: the cluster state connectivity $\partial_2$ remains
unchanged, and the added detectors are obtained as parities of
existing single-qubit $X$-measurement outcomes in the MBQC setting rather than
through additional high-weight stabilizer measurements. The
completed dual detectors define the extended boundary map
$\widetilde{\partial}_3$, and we denote the resulting fault
complex by
$\widehat{\operatorname{LP}}(A,1+x,1+x)$, whose dual detector
matrix is $\widetilde{\partial}_3^\mathsf T$.

For $f\in R$, write $\overline{f}=f(x^{-1})$.
After gauging out the short-cocycle sector, define
\(\Phi(a,c,d)=(\overline{d},\overline{a+c})\) for a completed
dual cocycle \(z=(a,c,d)\). The cocycle relation implies that
\(\Phi(z)\) is a cycle of
\(\operatorname{LP}(A^{\mathsf T},1+x)\). Moreover, for every
coboundary,
\begin{eqs}
\Phi\partial_2^\star(u,v,w)
=
\bigl((1+x)\overline{v+w},
A^{\mathsf T}\overline{v+w}\bigr),
\end{eqs}
which is a boundary of
\(\operatorname{LP}(A^{\mathsf T},1+x)\).

Moreover, if
$[\Phi(z)]=0$, then $z$ can be reduced to $(a,a,0)$. The
completed detector constraint gives
$K^\mathsf T a(1)=0$, implying
$a(1)\in\operatorname{im}A(1)^\mathsf T$ and hence that $(a,a,0)$ is
also a coboundary. Hence, $\Phi$ is a weight-nonincreasing
injection on the completed dual cohomology.
Since $A^\mathsf T$ is $\Delta$-limited and both
$A^\mathsf T$ and $A$ are $(\alpha,\beta)$-expanding, the
single-lifted-product distance bound
\cite{panteleev2021quantum}
applied to $A^\mathsf T$ gives
\begin{eqs}
d_1\!\left(\operatorname{LP}(A^\mathsf T,1+x)\right)
\geq \frac{\alpha}{4\Delta}\min\{\beta,1\}\ell \equiv \gamma_2 \ell.
\end{eqs}
Consequently, the fault distance of the detector-completed spacetime-lifted fault complex $\widehat{\mathrm{LP}}(A,1+x,1+x)$ is lower-bounded by $\gamma\ell$ where $\gamma=\min\{\gamma_1,\gamma_2\}$, indicating the claimed result.}
\end{proof}

\zw{\begin{theorem}\label{thm:DLP_PK_generalization}
For any $\ell>1$, let
$A\in\mathbb F_2[\mathbb Z_\ell]^{m\times n\Delta}$
be a $\Delta$-limited matrix, where
$n=\Theta(\log\ell)$
and
$m\leq n\Delta/2$, with the implicit constants and $\Delta$ independent of
$\ell$. Suppose that $A$ and $A^\mathsf T$ are
$(\alpha,\beta)$-expanding and that $A(1)$ has full row rank.
Then the detector-completed fault complex
$\widehat{\mathrm{LP}}(A,1+x,1+x)$ has fault distance
\begin{eqs}
d\geq\gamma\ell
 =\Omega\left(\frac{N}{\log N}\right),
\end{eqs}
where $\gamma$ is defined in
Theorem~\ref{thm:linear_distance} and 
$N=3\ell(n\Delta+m)$, 
composed of $\ell(n \Delta +2m)$ primal qubits and $\ell(2n \Delta  +m)$ dual qubits.
\end{theorem}}

\zw{The construction of Ref.~\cite[Proposition~1 and the proof of
Theorem~1]{panteleev2021quantum} provides $\Delta$-limited
matrices for which $A$ and $A^\mathsf T$ are expanding. The detector completion changes neither
the cluster state connectivity nor the number of cluster state
qubits.}

\section{Proof of fault distance for spacetime-lifted memory experiment}\label{appendix:memory_distance}
To prove that adding the initial and end time slices will not decrease the distance, we first define a bulk projection map which will be used in the proof of following theorems.

\zw{\begin{definition}[Bulk projection and inclusion]
\label{def:temporal_projection}
Let $R=\mathbb F_2[\mathbb Z_\ell]$ and $b=1+x$. For
convenience, denote the periodic and thickened fault complexes by
\begin{eqs}
\mathcal F_b=\operatorname{LP}(A,b,b),
\qquad
\mathcal F_C=\operatorname{LP}(A,b,C),
\end{eqs}
respectively. We denote their boundary maps by
$\partial_{k,b}$ and $\partial_{k,C}$, respectively.

On the primal space, the \emph{bulk projection} is
\begin{eqs}
P_1:
(u,u_\alpha,u_\beta;\,
v,v_\alpha,v_\beta;\,
w,w_\alpha,w_\beta)
\longmapsto
(u,v,w).
\end{eqs}
It retains only the bulk components. The corresponding inclusion
is \(I_1:(u,v,w)\mapsto
(u,0,0;\,v,0,0;\,w,w,w)\).

On the dual space, write the thickened blocks as
\((a_0,a_\alpha,a_\beta)\), \((c_0,c_\alpha,c_\beta)\), and
\((d_0,d_\alpha,d_\beta)\). The projection \(P_2\) retains the
first component of each block, so that
\(P_2z=(a_0,c_0,d_0)\), while the corresponding inclusion is
\(I_2:(a,c,d)\mapsto(a,0,0;\,c,c,c;\,d,d,d)\).
\end{definition}}

\begin{figure*}
    \centering
    \includegraphics[width=0.8\textwidth]{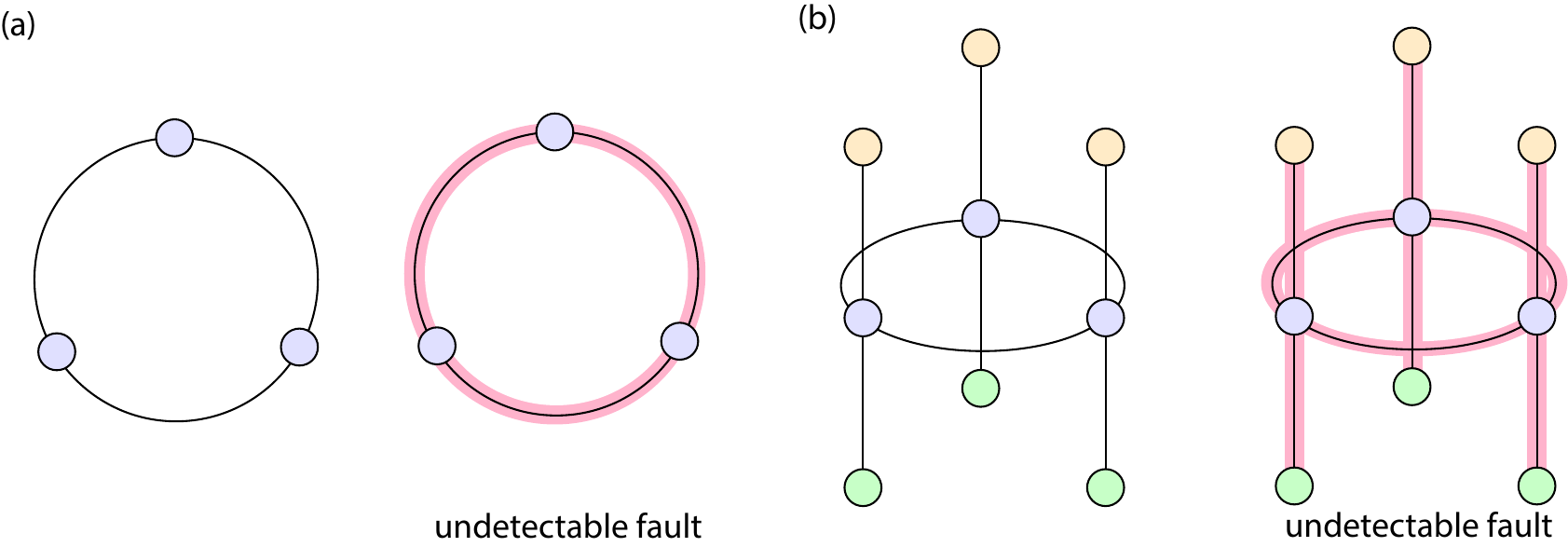}
    \caption{The third classical code $C$ used in the spacetime-lifted product $\text{LP}(A,B,C)$, where an edge represents a classical check, and a node represents a classical bit. (a) cyclic repetition code with parity check matrix $1+x$. The right figure represents the undetectable fault along the direction of $1+x$ which is proportional to $j_\ell$ with weight being at least $\ell$. (b) The `thickened' the cyclic repetition code by introducing the top and bottom portal whose parity check matrix is $C$. The right figure represents an undetectable fault which includes non-trivial action along the cyclic repetition component $1+x$, hence the undetectable fault has weight at least $3\ell$.}
    \label{fig:repetition_portal}
\end{figure*}

\zw{The key idea is that the temporal complex defined by $C$ is
chain-isomorphic to the direct sum of the periodic repetition
complex defined by $b=1+x$ and two contractible identity
complexes. Tensoring this decomposition with the first two
factors of the lifted product gives the projection and inclusion
maps $P$ and $I$ defined above. These maps induce inverse
isomorphisms on homology and cohomology, while their explicit
weight bounds allow us to transfer the distance of
$\mathcal F_b$ to $\mathcal F_C$ up to a constant factor.

It remains to extend this decomposition to the detector-completed
complexes. Let
$\widehat{\mathcal F}_b
=\widehat{\operatorname{LP}}(A,1+x,1+x)$
be the completed periodic fault complex from
Theorem~\ref{thm:linear_distance}. After binary expansion, let
$Y:G\to F_{2,b}$ denote its additional dual detector map, where
$G\cong\mathbb F_2^r$. Define
\begin{eqs}
Y_C=I_2Y,
\qquad
\widetilde{\partial}_{3,C}
=
\begin{pmatrix}
\partial_{3,C}&Y_C
\end{pmatrix}.
\end{eqs}
The chain map relation
$\partial_{2,C}I_2=I_1\partial_{2,b}$ gives
\begin{eqs}
\partial_{2,C}Y_C
=
\partial_{2,C}I_2Y
=
I_1\partial_{2,b}Y
=
0.
\end{eqs}
Therefore, $\widetilde{\partial}_{3,C}$ defines a valid detector-completed
thickened fault complex, which we denote by
$\widehat{\mathcal F}_C
=\widehat{\operatorname{LP}}(A,1+x,C)$.
This completion modifies only the dual detector map and leaves
the cluster state connectivity $\partial_{2,C}$ unchanged.

Furthermore, we show that temporal thickening preserves the
distances associated with both the first homology
$\ker(\partial_{1,C})/\operatorname{im}(\partial_{2,C})$
and the completed second cohomology
$\ker(\widetilde{\partial}_{3,C}^{\mathsf T})/
\operatorname{im}(\partial_{2,C}^{\mathsf T})$ of
$\widehat{\operatorname{LP}}(A,1+x,C)$.

\begin{theorem}
\label{thm:distance_double_LP_portal}
For a matrix $A$ satisfying the conditions of
Theorem~\ref{thm:linear_distance}, the detector-completed thickened fault
complex $\widehat{\operatorname{LP}}(A,1+x,C)$ has both
$1$-systolic and $2$-cosystolic distances at least
$\gamma\ell$. Consequently, replacing the temporal repetition
code $1+x$ by the thickened repetition code 
$C=
\begin{pmatrix}
1+x&0&0\\
1&1&0\\
1&0&1
\end{pmatrix}
$ 
does not change the fault distance scaling.
\end{theorem}
}

\zw{
\begin{proof}
Let $b=1+x$ and define
\begin{eqs}
V=
\begin{pmatrix}
1&0&0\\
1&1&0\\
1&0&1
\end{pmatrix},
\qquad
V^{-1}=V,
\qquad
CV=\operatorname{diag}(b,1,1).
\end{eqs}
Hence the temporal complex defined by $C$ is chain-isomorphic
to the direct sum of the periodic repetition complex defined by
$b$ and two contractible identity complexes.

Tensoring this decomposition with the identity acting on
$\operatorname{LP}(A,b)$ gives chain maps $P$ and $I$, whose
degree-$1$ and degree-$2$ components are the projection and
inclusion maps defined above. They satisfy
\begin{eqs}
\partial_{k,C}I_k
=
I_{k-1}\partial_{k,b},
\qquad
P_{k-1}\partial_{k,C}
=
\partial_{k,b}P_k,
\qquad
P_kI_k
=
\operatorname{id}.
\end{eqs}
Since the other two summands are contractible, $P$ and $I$
induce inverse isomorphisms on homology and cohomology.

On the primal space, $P_1$ is weight non-increasing, while
$I_1$ increases the weight by at most a factor of three. Since
detector completion changes only the degree-$3$ boundary map, it
leaves the primal homology unchanged. Therefore,
\begin{eqs}
d_1\!\left(\widehat{\operatorname{LP}}(A,b,b)\right)
\leq
d_1\!\left(\widehat{\operatorname{LP}}(A,b,C)\right)
\leq
3d_1\!\left(\widehat{\operatorname{LP}}(A,b,b)\right).
\end{eqs}

It remains to check that the detector completion is compatible
with this decomposition. By construction, $Y_C=I_2Y$, and
\begin{eqs}
\partial_{2,C}Y_C
=
\partial_{2,C}I_2Y
=
I_1\partial_{2,b}Y
=
0.
\end{eqs}
Moreover, $P_2Y_C=P_2I_2Y=Y$. Thus $P$ and $I$ extend to the
detector-completed complexes by acting identically on the added detector
generators. Since $Y_C$ lies entirely in the periodic summand,
the other two summands remain contractible.

The transpose map $I_2^{\mathsf T}$ is weight non-increasing,
while $P_2^{\mathsf T}$ is a weight-preserving coordinate
inclusion. Since they induce inverse isomorphisms on the
completed second cohomology,
\begin{eqs}
d^2\!\left(\widehat{\operatorname{LP}}(A,b,C)\right)
=
d^2\!\left(\widehat{\operatorname{LP}}(A,b,b)\right).
\end{eqs}
Theorem~\ref{thm:linear_distance} therefore gives
\begin{eqs}
d_1\!\left(\widehat{\operatorname{LP}}(A,b,C)\right)
\geq\gamma\ell,
\qquad
d^2\!\left(\widehat{\operatorname{LP}}(A,b,C)\right)
\geq\gamma\ell,
\end{eqs}
which proves the claim.
\end{proof}
}

This fault complex \textcolor{black}{$\widehat{\mathrm{LP}}(A,1+x, C)$} exhibits only membrane-like dual logical correlations of the form $(u,\mathbf{0},\mathbf{0},u,u,u,v,v,v)$ where $u,v$ satisfy $Au + (1+x)v= \mathbf{0}$. The corresponding dual logical correlation operator takes the form $X(u,v)^{\otimes 3} \cdot X(u,\mathbf{0},\mathbf{0},\mathbf{0},\mathbf{0},\mathbf{0},\mathbf{0},\mathbf{0},\mathbf{0})$, where $X(u,v)^{\otimes 3}$ correlates the three temporal layers. Although this is an $X$-type operator acting on dual qubits, its support $(u,v)$ coincides with the support of the logical $\overline{Z}$ operators of the lifted product code $\text{LP}(A,1+x)$. This reflects the correspondence between primal/dual sectors in the fault complex picture and $X/Z$ sectors in the circuit picture. Consequently, the protocol can be interpreted as a spacetime-lifted memory experiment in which all qubits are initialized in the $Z$ basis and finally measured in the $Z$ basis.

\zw{\textbf{Parameters of the spacetime-lifted fault complex.}
Similar to Theorem~\ref{thm:DLP_PK_generalization}, the
detector-completed fault complex
$\widehat{\operatorname{LP}}(A,1+x,C)$ has almost-linear fault
distance. The number of primal qubits is
$3\ell(n\Delta+2m)$, while the number of dual qubits is
$3\ell(2n\Delta+m)$. Thus the total number of qubits in the
fault-tolerant cluster state is
\begin{eqs}
N=9\ell(n\Delta+m),
\end{eqs}
with each input/output portal having size
$\ell(n\Delta+m)$. Combining
$\log N=\Theta(\log\ell)$, $n=\Theta(\log N)$, and
Theorem~\ref{thm:distance_double_LP_portal}, we find that the
number of encoded logical qubits is $K=\Theta(\log N)$ and that
the fault distance satisfies
\begin{eqs}
d\geq\gamma\ell
=
\Omega\left(\frac{N}{\log N}\right).
\end{eqs}}

\section{1D cluster state example for logical teleportation}\label{appendix:1D_cluster}

In this appendix, we use 1D cluster state as a minimal example to demonstrate that a fault complex can support logical correlation even without the presence of nontrivial dual correlations. This clarifies the distinction between the logical $\overline{Z_{\text{in}}Z_{\text{out}}}$ correlation required for teleportation and the notion of dual correlation in the fault complex formalism, thereby reconciling our conditions for logical teleportation. For simplicity, we consider the protocol on a 1D cluster state along a line whose two boundary qubits are entangled after measuring the bulk, although the same analysis extends directly to general fault complexes for logical teleportation of encoded logical qubits.

\begin{example}[1D cluster state]
    We consider a length-$(2L+1)$ 1D cluster state with fault complex 
    \[
\begin{tikzcd}[row sep=0.7em, column sep=1em]
F_2 \arrow[r, "\partial_2"] &
F_1,  \\
{\scriptstyle \text{dual fault position}} &
{\scriptstyle \text{primal fault position}} 
\end{tikzcd}
\]
which is a 2-term complex, because there is no primal and dual detectors. This process is a foliation of a single physical qubit, which has no nontrivial parity-check matrices $H_X$ or $H_Z$.

We assign qubits on odd sites to the primal sector $F_1$ and qubits on even sites to the dual sector $F_2$. The boundary maps $\partial_3$ and $\partial_1$ are zero maps and trivial, while $\partial_2$ specifies the connectivity of the cluster state:
\begin{eqs}
    \partial_2=\begin{pmatrix}
        1 & 0& 0& \cdots & 0& 0\\
        1 & 1& 0& \cdots & 0& 0\\
        0 & 1& 1& \cdots & 0& 0\\
        \vdots & \vdots & \vdots & \ddots & \vdots& \vdots \\
        0 & 0& 0& \cdots & 1& 1\\
        0 & 0& 0& \cdots & 0& 1
    \end{pmatrix},
\end{eqs}
which is the transpose of the parity-check matrix of a length-$(L+1)$ repetition code with open boundary conditions. We take $\partial_2$ to be the transpose of the repetition-code matrix so that the two boundary qubits belong to the primal sector. If instead one uses the standard repetition-code parity-check matrix, the resulting teleportation protocol would involve two boundary dual qubits; the two choices are equivalent up to swapping the roles of primal and dual qubits.

Such boundary map is a  $(L+1) \times L$ binary matrix where each row corresponds to a primal qubit, and each column corresponds to a dual qubit. The first and last rows contain only a single nonzero entry because the boundary primal qubits each couple to only one neighboring dual qubit.

The corresponding cluster state stabilizer group is
\begin{eqs}
    \langle X_1 Z_2, Z_{i-1} X_{i} Z_{i+1}, Z_{2L} X_{2L+ 1} \rangle,\quad \forall 2\leq i\leq 2L,
\end{eqs}
which can be written as a more compact form $\langle X_a Z_{\partial_2 a}, X_b Z_{\partial_2^\mathsf T b} \rangle $ for all $a \in F_2, b \in F_1$ as \eqref{eq:cluster_stabilizer}. Note that the system contains $2L+1$ qubits in total. Here $a$ and $\partial_2^\mathsf T b$ are $L$-dimensional binary vectors associated with the dual qubits, while $b$ and $\partial_2 a$ are $(L+1)$-dimensional binary vectors associated with the primal qubits. When used as subscripts of Pauli operators, these vectors act separately on the primal and dual sectors. For example, the compact form of cluster state stabilizer should read as
\begin{eqs}
X_a Z_{\partial_2 a}
:=
X_{a,\text{dual}}
\otimes
Z_{\partial_2 a,\text{primal}},\quad \forall  a \in F_2\\
X_b Z_{\partial_2^\mathsf T b}
:=
X_{b,\text{primal}}
\otimes
Z_{\partial_2^\mathsf T b,\text{dual}},\quad \forall b \in F_1
\end{eqs}
where the additional labels “primal” and “dual” indicate the corresponding subsets of qubits.

We choose the first and last qubits as the two portals $Q_{\text{in}},Q_{\text{out}}$, and we denote $X_{\text{in}}:=X_1, Z_{\text{in}}:=Z_1, X_{\text{out}}:= X_{2L+1}, Z_{\text{out}}:= Z_{2L+1}$. All other qubits are measured in the single-qubit $X$ basis and give measurement outcome $m_i=\pm 1$ for $2 \leq i\leq 2L$. Indeed, most cluster state stabilizers anticommute with the single-qubit measurements and are therefore destroyed, leaving only two independent stabilizers that generate the Bell correlations:
\begin{eqs}
\langle \varphi_X X_{\text{in}} X_{\text{out}},  \varphi_Z Z_{\text{in}} Z_{\text{out}}\rangle 
\end{eqs} 
where $\varphi_X=\prod_{  \substack{3 \leq i \leq 2L-1\\ i \in \text{odd}}}m_{i}$ and $\varphi_Z=\prod_{ i\in \text{even}} m_{i}$ are $\pm 1$ signs given by the measurement outcomes. The logical $X_{\text{in}}X_{\text{out}}$ correlation originates from the primal correlation corresponding to the nontrivial element of $\ker (\partial_2^\mathsf T)$, namely the all-one vector \textcolor{black}{$b=\mathbf{1}_{L+1}$}. Explicitly, the all-one vector \textcolor{black}{$b=\mathbf{1}_{L+1}$} means taking the product of the cluster state stabilizers  supported on primal qubits gives
\begin{eqs}
    X_b Z_{\partial_2^\mathsf T b}= \bigotimes_{i \in \text{odd}}X_i.
\end{eqs}
After measuring the bulk primal qubits, the measured $X_i$ operators are replaced by their outcomes $m_i$, leaving the effective logical correlation
\begin{eqs}
    \bigotimes_{i \in \text{odd}}X_i \rightarrow \varphi_X X_{\text{in}} X_{\text{out}}.
\end{eqs}
The logical $\varphi_Z Z_{\text{in}}Z_{\text{out}}$ correlation arises from cluster state stabilizers \textcolor{black}{$X_a Z_{\partial_2 a}$ with $a=\mathbf{1}_{L} \notin \mathrm{ker}(\partial_2)$}. Such cluster state stabilizer commute with all the single-qubit $X$-measurement within the bulk and is a product of $X$ operator acting on bulk dual qubits together with a $Z_1Z_{2L+1}$ part supported on the two portals which is a product of the cluster state stabilizers $Z_{i-1} X_i Z_{i+1}$ centered on all the dual qubits (even sites $i$)
\begin{eqs}
X_{a}  Z_{\partial_2 a}= Z_{\text{in}} Z_{\text{out}} \bigotimes_{i \in \text{even}} X_i.
\end{eqs}
The key point is that the $Z$ component of this stabilizer has support only on the two portals, since
\begin{eqs}
\partial_2 \mathbf{1}_L= \begin{pmatrix}
    1 & 0& 0& \cdots & 1
\end{pmatrix}
\end{eqs}
only has $1$ on the first and last column and gives the $Z_1 Z_{2L+1}$ acting on the portals, which are the first and last primal qubits respectively.
After measuring all the dual qubits, the measured $X_i$ operators for even $i$ are replaced by their outcomes $m_i$, which gives the logical correlation
\begin{eqs}
    Z_{\text{in}}Z_{\text{out}} \bigotimes_{i \in \text{even}} X_i \rightarrow \varphi_Z Z_{\text{in}}Z_{\text{out}},
\end{eqs}
matching the conditions we stated in the main text.

We emphasize that $ZZ$ correlation is different from the dual correlation of the fault complex. Such fault complex does not possess any nontrivial dual correlation. Because $\partial_2$ is a $(L+1)\times L$ full-row-rank matrix, indeed, $\ker (\partial_2)=\{\mathbf{0}_{L}\}$ is trivial. This reflects the fact that dual correlations in the fault complex formalism correspond to pure-$X$ operators supported solely on dual qubits in the cluster state stabilizer group, but such operator doesn't exist in this example. On the other hand, for logical teleportation, we should look for logical correlations that are in the cluster state stabilizer group and can be written as $\overline{X_{\text{in}}X_{\text{out}}}$ and $\overline{Z_{\text{in}}Z_{\text{out}}}$ acting on the input and output portals, and Pauli $X$ acting on the measured bulk part. 

This analysis holds true for general fault complexes, including cases where the detector matrices $\partial_1$ and $\partial_3^\mathsf T$ are nontrivial. For example, foliated fault complexes of CSS codes support the logical $\overline{X_{\text{in}}X_{\text{out}}}$ and $\overline{Z_{\text{in}}Z_{\text{out}}}$ correlations required for fault-tolerant logical teleportation. However, because the temporal repetition code appearing in the foliated construction has open boundary conditions, the resulting fault complex cannot simultaneously possess both nontrivial primal and dual correlations between portals.
\end{example}

\end{document}